\documentclass[conference]{IEEEtran}
\IEEEoverridecommandlockouts
\usepackage{cite}
\usepackage{amsmath,amssymb,amsfonts,amsthm}
\usepackage{algorithmic,algorithm}
\usepackage{graphicx}
\usepackage{textcomp}
\def\BibTeX{{\rm B\kern-.05em{\sc i\kern-.025em b}\kern-.08em
    T\kern-.1667em\lower.7ex\hbox{E}\kern-.125emX}}
\newtheorem{theorem}{Theorem}
\newtheorem{definition}{Definition}

\usepackage{hyperref}

\usepackage[dvipsnames]{xcolor}

\begin{document}

\title{Differentially Private ANOVA Testing
}

\author{\IEEEauthorblockN{Zachary Campbell}
\IEEEauthorblockA{\textit{Mathematics Department} \\
\textit{Reed College}\\
Portland, OR, USA \\
campbza@reed.edu}
\and
\IEEEauthorblockN{Andrew Bray}
\IEEEauthorblockA{\textit{Mathematics Department} \\
\textit{Reed College}\\
Portland, OR, USA \\
abray@reed.edu}
\and
\IEEEauthorblockN{Anna Ritz}
\IEEEauthorblockA{\textit{Biology Department} \\
\textit{Reed College}\\
Portland, OR, USA \\
aritz@reed.edu}

\and
\IEEEauthorblockN{Adam Groce}
\IEEEauthorblockA{\textit{Mathematics Department} \\
\textit{Reed College}\\
Portland, OR, USA \\
agroce@reed.edu}
}

\maketitle

\begin{abstract}
Modern society generates an incredible amount of data about individuals, and releasing summary statistics about this data in a manner that provably protects individual privacy would offer a valuable resource for researchers in many fields.  We present the first algorithm for analysis of variance (ANOVA) that preserves differential privacy, allowing this important statistical test to be conducted (and the results released) on databases of sensitive information.  In addition to our private algorithm for the $F$ test statistic, we show a rigorous way to compute $p$-values that accounts for the added noise needed to preserve privacy.  Finally, we present experimental results quantifying the statistical power of this differentially private version of the test, finding that a sample of several thousand observations is sufficient to detect variation between groups. The differentially private ANOVA algorithm is a promising approach for releasing a common test statistic that is valuable in fields in the sciences and social sciences.\end{abstract}

\begin{IEEEkeywords}
differential privacy, data privacy, statistics, ANOVA
\end{IEEEkeywords}

\section{Introduction}
The goal of private data analysis is to gain useful information from databases of personal information while at the same time protecting the privacy of the individuals in the database.  This has traditionally been done using ad hoc ``anonymization'' techniques, but in recent years these techniques have failed on numerous occasions (ex., \cite{sweeney2002k,narayanan2008robust, homer2008resolving}).  Differential privacy \cite{dwork2006calibrating} has emerged as a rigorous framework for providing provably strong privacy while publicly disseminating useful information about a data set.

Differential privacy is not a method for producing publishable output, but rather a \textit{sufficient condition} that guarantees privacy is protected with the level of protection parameterized by a parameter $\epsilon$.  In order to satisfy differential privacy, the information being published (generally summary statistics) must in some way be randomized.  In the simplest query algorithms, random noise is added to the correct (but not necessarily private) output.  The result is often a good estimate of the true value.  Some schemes are more complicated, adding noise in various ways \textit{during} a computation.  Differentially private algorithms have been created to allow for the release of simple summary statistics like means \cite{dwork2006calibrating} and medians \cite{dwork2009differential, nissim2007smooth}, as well as more elaborate statistical output, including histograms \cite{dwork2006calibrating}, linear regressions \cite{dwork2009differential, zhang2012functional, chaudhuri2009privacy}, chi-squared tests \cite{vu2009differential} and a variety of machine learning techniques. 

In this paper, we present the first differentially private algorithm for conducting a one-way analysis of variance (ANOVA).  This is a statistical test for settings with a categorical independent variable and a continuous dependent variable.  It is a very widely used piece of statistical analysis and is used as the primary test of significance for many studies in biology, medicine, and the social sciences.

A one-way ANOVA outputs a test statistic called an $F$ statistic, which measures the ratio of the variation in means of categories to variation of individual observations.  Our algorithm outputs an estimate of this value.  However, $F$ statistics are not of interest directly, but rather as a tool for measuring statistical significance of the result.  One cannot simply look up our noisy estimate of $F$ in an $F$-distribution table to find a $p$-value, because the noise makes extreme values more likely than they would otherwise be.

Our second contribution is to give a rigorous method for computing $p$-values.  This means that our algorithm produces final results that can be easily interpreted.  (Most existing algorithms do not output these sorts of results.  For example, algorithms for linear regression output only the best fit line, which is useful for making predictions but not sufficient for testing independence of two variables.)

Finally, we run a series of experiments to measure the statistical power of our algorithm.  That is, for a given privacy parameter and effect size, we show the required size of the database for the algorithm to indicate a statistically significant result.  This provides a quantitative analysis of the utility-privacy tradeoff under our algorithm and provides a basis against which to compare future algorithms that might seek to improve upon this work. Our code is freely available at {\small {\color{blue}{\texttt{\url{https://github.com/campbza/Differentially-private-ANOVA}}}}.

%
%
%
%

\section{Background}

Below we provide technical background on the two main topics of our paper, differential privacy and the ANOVA statistical test.

\subsection{Differential Privacy}

Imagine your data is part of a database about which some information is released.  Intuitively, differential privacy guarantees that the published output would be similarly likely regardless of the specific values of your data.  If the output is similarly likely regardless of your data's value, then someone seeing that output is unable to meaningfully infer anything about you.  (This inability to learn from the output has been formalized and proven to be a consequence of differential privacy.  See \cite{kasiviswanathan2008note} for more details.)  Differential privacy was introduced by Dwork et al.~in 2006, and the core definition and theorems presented below were all present in that initial work. \cite{dwork2006calibrating}

Formally, we consider two databases $D$ and $D'$, each consisting of $n$ rows.  (A ``row'' consists of all information connected to a given record.)  We say $D$ and $D'$ are ``neighboring'' if they are the same except for one altered row.  A function $f$ on the database is said to be \textit{differentially private} if the probability of any output of $f$ is roughly the same for any neighboring inputs $D$ and $D'$.

\begin{definition}[Differential Privacy]
A (randomized) algorithm $f$ with range $R$ is $\epsilon$-differentially private if 
for all $S \subseteq R$ and for all neighboring databases $D$ and $D'$ 
\[
\Pr[f(D) \in S] \leq e^\epsilon \Pr[f(D')\in S].
\]
\end{definition}

We call $\epsilon$ the \textit{privacy parameter}, which is chosen by the user.  Differential privacy has many useful properties.  Some, like those mentioned above, guarantee that it really is sufficient to guarantee privacy.  Others allow for easier creation of private algorithms and easier use of those algorithms in practice.  Two properties in particular will be useful to us here.
\begin{theorem}[Composition]\label{thm:composition}
If $f$ is $\epsilon_1$-differentially private and $g$ is $\epsilon_2$-differentially private, then if $h$ simply returns the outputs of both $f$ and $g$ (i.e., $h(D) = (f(D), g(D))$) then $h$ is $(\epsilon_1+\epsilon_2)$-differentially private.
\end{theorem}
We will use this property to combine algorithms for intermediate calculations into an algorithm for ANOVA as a whole.  It also means that our ANOVA algorithm can be combined with other differentially private algorithms to conduct larger, more complex statistical analyses.

The other useful property is resistance to post-processing:
\begin{theorem}[Post-processing]
If $f$ is $\epsilon$-differentially private and $g$ is an arbitrary function, then if $h = g \circ f$ (i.e., $h(D) = g(f(D))$) then $h$ is also $\epsilon$-differentially private.
\label{thm:postprocessing}
\end{theorem}
This property is primarily a sanity check on the definition, since a good definition of privacy must have this property.  In our case, it also means that our process for converting a noisy $F$ statistic into a $p$-value maintains privacy.

We build our algorithm by using the standard Laplacian mechanism for individual components and then combining the result.  This technique relies on knowing the \textit{sensitivity} of certain computations.  This sensitivity is defined as the maximum effect that can occur when a single row is changed.
\begin{definition}[Sensitivity]
The sensitivity of a (deterministic) function $f$ on databases with real number outputs is the maximum over neighboring databases $D$ and $D'$ of $|f(D)-f(D')|$.
\end{definition}

We also define the Laplace distribution, for which this technique is named:
\begin{definition}
The \emph{Laplace distribution} (centered at 0) with scale $b$ is denoted $\text{Lap}(b)$ and has probability density function 
\[
f(x) = e^{-|x|/b}/2b.
\]
\end{definition}

We can now give the Laplacian mechanism, which gives a differentially private algorithm for any function with a known bound on sensitivity.
\begin{theorem}[Laplacian mechanism]\label{thm:lapmech}
Let $f$ be a function with sensitivity at most $s$.  Let $L$ be a random variable drawn from $\text{Lap}(s/\epsilon)$.  Then the function $f'(D) = f(D) + L$ is $\epsilon$-differentially private.
\end{theorem}

\subsection{ANOVA Testing}

A one-way analysis of variance (ANOVA) evaluates whether the data is consistent with a null hypothesis under which all groups of the categorical variable share the same mean for some continuous response variable \cite{snedecorcochran}.  For example, it might test whether a choice of treatment (ex., surgery, medication, or none) had an effect on the mean life expectancy of patients with a particular disease. Deviation from the null hypothesis is measured with an $F$ statistic. If the resulting $F$ statistic is sufficiently unlikely to occur under the null hypothesis, the analyst can safely conclude that average response is not independent of group membership.


A database $D$ for an ANOVA test contains $n$ rows, each containing a categorical variable that takes on one of $k$ allowed values and a continuous variable.  (We will throughout this paper assume that the continuous variable has been normalized and has allowed range $[0,1]$.)  We let $y_{ij}$ represent the $j^\text{th}$ row in category $i$.  We let $\overline{y_i}$ represent the mean of group $i$, and $\overline{y}$ represent the grand mean (i.e, the mean of all values).  We let $n_i$ be the number of values in group $i$.

An ANOVA test produces an $F$ statistic.  This can be calculated from two intermediate values.  First is SSA, the sum of the squared error of all sample means compared to the grand mean weighted by the size of each group.  This measures the variation between the means of each group.
\[
\text{SSA}(D) = \sum_{i=1}^{k} n_i(\overline{y}_i - \overline{y})^2.
\]

The second intermediate value is SSE, the sum squared error of all values compared to the grand mean.  This measures the variation among the data points as a whole. 
\[
\text{SSE}(D) = \sum_{i=1}^{k} \sum_{j=1}^{n_i} (y_{ij} - \overline{y}_i)^2.
\]

The $F$ statistic compares these two values.  Intuitively, if all groups are distributed identically, then the random variation between groups should scale proportionately to the variation in individual values.  If the groups have different means, SSA will grow faster than SSE.  Therefore $F$ is a ratio of SSA to SSE, with each adjusted to account for their respective degrees of freedom.
\begin{equation*}
F(D) = \frac{\text{SSA}(D)/(k-1)}{\text{SSE}(D)/(n-k)}
\end{equation*}

%

When sample sizes are large, the distribution of the $F$ statistic under the null hypothesis (given known values of $k$ and $n$) is well-known.  Given a particular $F$ statistic, one can compute the probability that a value that high or higher would have occurred by random chance if the null hypothesis was true.  This is the $p$-value.  A low $p$-value (often 0.05 or less) is seen as ``significant'' and is taken as reasonable evidence that the null hypothesis can be rejected.

\section{Our Algorithm}
Our algorithm uses the Laplacian mechanism defined previously. In order to do so, we first prove sensitivity bounds on our queries.  We assume all data has been normalized so that $y_{ij}$ values fall in $[0,1]$ and that the number of possible groups $k$ is fixed and public, but the number of samples from each group is not known.  The most straightforward approach would be to bound the sensitivity of $F$, but this is very difficult.  Instead we bound the sensitivity of SSA and SSE.

\begin{theorem}\label{thm:SSEbound}
	Recall that 
	\[
		\text{SSE}(D) = \sum_{i=1}^{k} \sum_{j=1}^{n_i} (y_{ij} - \overline{y}_i)^2.
	\]
	SSE has sensitivity bounded above by 7.
\end{theorem}

\begin{proof}
We must analyze the effect of changing a single data point on SSE.  There is a term in the sum for each data point, each with a value in $[0,1]$.  The term corresponding to the changed data point might change arbitrarily, so we bound this term's contribution to the change by 1.

Next we look at the effect that change has on the other terms through its effect on $\overline{y_i}$ and $\overline{y_j}$, where the data point's change moves it from group $i$ to group $j$.  (This is the worst case --- changes that don't change the rows groups would have lesser effect.)  Let $n_i$ and $n_j$ be the number of terms (i.e., data points) in each group excluding the data point being changed.  We can bound the change's effect on $\overline{y_i}$ by $1/n_i$.  We are are therefore changing the value being squared in each term, $y_{ij} - \overline{y}_i$, by at most $1/n_i$.  In general, if you change $a^2$ to $(a+b)^2$, the change is $2ab + b^2$.  Here we have $a = y_{ij} - \overline{y}_i \leq 1$, so we can bound this by $2b+b^2$.  Plugging in $b=1/n_i$ gives us a per-term change of
\begin{equation*}
2/n_i + 1/n_i^2.
\end{equation*}
We then have $n_i$ terms affected by this change, for a total effect of $2 + 1/n_i < 3$.

	Similarly, the effect on the terms in group $j$ is bounded by 3.  This gives us a change of at most 1 for term corresponding to the row being changed and a change of at most 3 for each of the two groups of other terms this row affects, for a total sensitivity of 7.
\end{proof}

%

\begin{theorem}\label{thm:SSAbound}
	Recall that
	\[
		\text{SSA}(D) = \sum_{i=1}^{k} n_i(\overline{y}_i - \overline{y})^2.
	\]
	The sensitivity of SSA is bounded above by $9 + 5/n$.
\end{theorem}

\begin{proof}
This proof follows the same logic as the previous one.  Instead of thinking of $k$ terms, each weighted by $n_i$, think of the sum as $n$ terms, with each group resulting in $n_i$ identical terms:
\begin{equation*}
g(d) = \sum_{i=1}^{k} \sum_{j=1}^{n_i} (\overline{y}_i - \overline{y})^2
\end{equation*}
Again, the term corresponding to our changed data point can change by 1.  Other terms in group $i$ see $\overline{y_i}$ changed by $1/n_i$ and $\overline{y}$ changed by $1/n$.  Again this is a change from $a^2$ to $(a+b)^2$ but this time $b=1/n_i + 1/n$.  This gives us a per-term change of:
\begin{align*}
&2(1/n_i + 1/n) + (1/n_i + 1/n)^2 \\
= &2/n_i + 2/n + 1/n_i^2 + 2/nn_i + 1/n^2
\end{align*}
We then sum this per-term change over $n_i$ terms to get a change of
\begin{eqnarray*}
2 + 2n_i/n + 1/n_i + 2/n + n_i/n^2 \\
\leq 3 + 2n_i/n + 2/n + n_i/n^2
\end{eqnarray*}
There is an identical bound on the total change in the terms representing group $j$, the group the data is moving to.  Then we consider the other terms.  Say $n^*$ represents the number of terms not in either of these groups.  Each of those terms sees $\overline{y}$ change by at most $1/n$.  With $b=1/n$ we get a total change of $2n^*/n + n^*/n^2$.

We then add together those four separate bounds: the bound of 1 on the term's own change; the bound on the total change of terms in group $i$; the equivalent bound on change in group $j$; and the bound on the rest of the terms.
\begin{align*}
1 + (3 + 2n_i/n + &2/n + n_i/n^2) + \\
&(3 + 2n_j/n + 2/n + n_j/n^2) + \\
&2n^*/n + n^*/n^2 \\
&\leq 7 + \frac{2(n_i + n_j + n^*)}{n} + 4/n + \frac{n_i + n_j + n^*}{n^2}
\end{align*}
Given that $n_i+n_j+n^* < n$, this is bounded by
\begin{equation*}
9 + 5/n.
\end{equation*}
\end{proof}


Using these sensitivity bounds it is straightforward to construct our algorithm. We compute estimates $\widehat{\text{SSE}}$ and $\widehat{\text{SSA}}$ of SSE and SSA using the Laplacian mechanism, then combine the results to estimate the $F$ statistic.  See Algorithm 1 for formal details.

\begin{algorithm}\label{alg:Fhat}
	\caption{Differentially private ANOVA}
	\begin{algorithmic}
		\STATE \textbf{Input:} Database $D$, $\epsilon$ value
		\STATE Compute $\widehat{\text{SSA}} = \text{SSA}(D) + Z_1$ where $Z_1\sim\text{Lap}\left(\frac{9 + 5/n}{\epsilon/2}\right)$
		\STATE Compute $\widehat{\text{SSE}} = \text{SSE}(D) + Z_2$ where $Z_2\sim\text{Lap}(7/(\epsilon/2))$
		\STATE Compute $\widehat{F} = \frac{\widehat{\text{SSA}}/(k-1)}{\widehat{\text{SSE}}/(n-k)}$
		\STATE \textbf{Output:} $\widehat{F}, \widehat{\text{SSA}}, \widehat{\text{SSE}}$
	\end{algorithmic}
\end{algorithm}

\noindent We now prove that this algorithm is $\epsilon$-differentially private.

\begin{theorem}{}
	Algorithm 1 is $\epsilon$-differentially private.  
\end{theorem}

\begin{proof}
By the sensitivity bounds of Theorems \ref{thm:SSEbound} and \ref{thm:SSAbound} and the known Laplacian mechanism (Theorem \ref{thm:lapmech}) we know that the computations of $\widehat{\text{SSA}}$ and $\widehat{\text{SSE}}$ are each $\epsilon/2$-differentially private.  By our composition theorem (Theorem \ref{thm:composition}), outputting both is $\epsilon$-differentially private.  Finally, since the computation of $\widehat{F}$ does not require access to the database, Theorem \ref{thm:postprocessing} shows that it can be added to the output without loss of privacy.
\end{proof}

Note that while $\widehat{F}$ is the value of most immediate interest, we can also output $\widehat{\text{SSA}}$ and $\widehat{\text{SSE}}$ without additional privacy loss.  These values can sometimes give other useful information about the database ``for free.''  For example, $\widehat{\text{SSE}}/(n-k)$ is an estimate of the variance of the data, which might be of independent interest.



\subsection{Computing $p$-values}
\label{sec:pvals}
Normally in an ANOVA test we get a value for $F$.  We can then refer to the known $F$-distribution to find the $p$-value, the probability that the observed value, or more extreme, would have occurred by chance under the null hypothesis. The higher the 
$F$-ratio, the lower the $p$-value. However, our 
algorithm returns an approximation, $\widehat{F}$, of the real $F$ value.  One could simply compare this value to the $F$ distribution and get an approximation of the true $p$-value, but one could also be more thorough and actually compute a distribution for $\widehat{F}$ values under the null hypothesis and compare to \textit{that} distribution.  We do the latter.

Under the null hypothesis, SSA is drawn from $\sigma^2\chi^2_{k-1}$, the chi-squared distribution with $k-1$ degrees of freedom scaled by the variance $\sigma^2$.  SSE is similarly drawn from $\sigma^2\chi^2_{n-k}$.  In the classic setting, the $\sigma^2$ factors cancel in the ratio, making the $F$ distribution independent of the variance of the underlying data points.

This is not true in our scenario.  $\widehat{SSE}$ is drawn from $\sigma^2\chi^2_{n-k} + \text{Lap}(7/(\epsilon/2))$.  The added Laplacian noise makes the resulting ratio dependent on $\sigma^2$.  Luckily, $\widehat{SSE}/(n-k)$ is an estimate of $\sigma^2$, and we can use this estimate to accurately compute a distribution for $\widehat{F}$.

To simulate a null hypothesis $\widehat{F}$ distribution we choose $\widehat{SSE}$ and $\widehat{SSA}$ from the above distributions, using $\widehat{SSE}/(n-k)$ as our estimate of $\sigma^2$.  We repeat this a large number of times (100,000 in our experiments).

When we see an output of Algorithm 1, we use the reported $\widehat{SSE}$ value to compute the above distribution for $\widehat{F}$ under the null hypothesis.  Then we look at the percentage of that distribution falling above the particular value we saw for $\widehat{F}$ and return that as our $p$-value.  Note that this entire computation of $p$-value is done independently of the database, using only the output of Algorithm 1.  Therefore it can be added to the algorithm as an additional output with no added privacy cost.

One of our key findings is that this added care in computing the $p$-value is absolutely necessary.  Consider the distributions shown in Figure~\ref{fig:fdist}.  Here we have shown the standard $F$ distribution when $n=10,000$ and $k=10$, along with distributions of the estimate $\widehat{F}$ when $\epsilon = 1$ and when $\epsilon=0.1$.  The dotted line shows the value of the $F$ statistic that would normally indicate statistical significance with $p=0.05$.  However, at this point with $\epsilon=1$ the true $p$-value we calculate is 0.13.  With $\epsilon=0.1$ the true $p$-value at that point is 0.41.

\begin{figure}
\centering
\includegraphics[width=\linewidth]{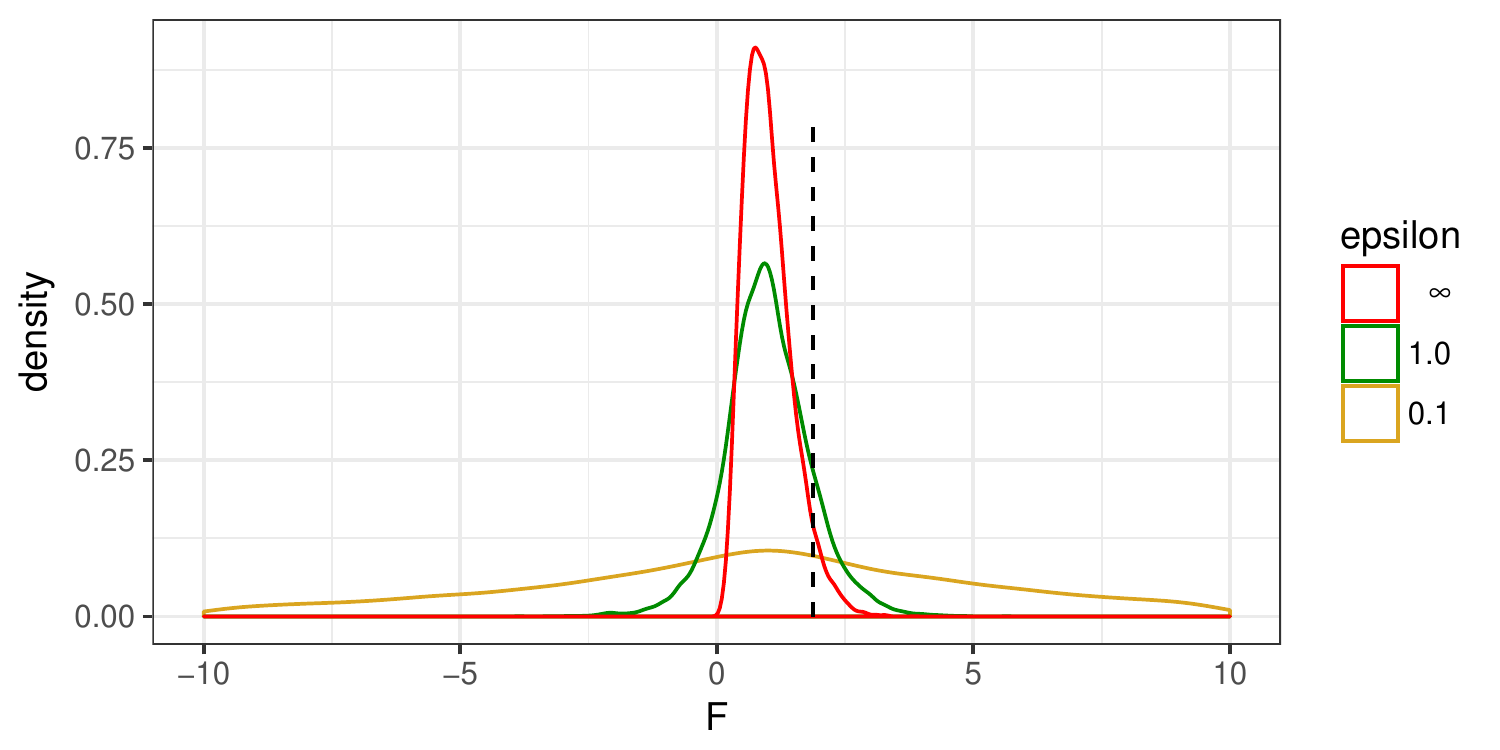}
\caption{The distribution of the $F$ statistic (red) overlaid with distributions of the estimate $\widehat{F}$ for two values of $\epsilon$ (green and gold).  In all cases, $n=10,000$ and $k=10$.  The significant difference between these curves indicates that even though $\widehat{F}$ is an estimate of $F$, ignoring the added noise in $\widehat{F}$ produces unacceptably inaccurate conclusions.}
\label{fig:fdist}
\end{figure}

\section{Results}

We assessed the power of our differentially-private ANOVA test on synthetic data as we increased the database size $n$ for different privacy parameters $\epsilon$.  We generated normally-distributed data from three equally-sized groups with the same standard deviation and three different means.  That is, the $n$ values were comprised of $n/3$ values drawn from $\mathcal{N}(0.35,0.15)$, $n/3$ values drawn from $\mathcal{N}(0.5,0.15)$, and $n/3$ values drawn from $\mathcal{N}(0.65,0.15)$.  Values were truncated to be within $[0,1]$.  We calculated estimates $\widehat{SSA}$, $\widehat{SSE}$, and $\widehat{F}$, and computed the $p$-value as described in Section~\ref{sec:pvals}.  We repeated this procedure (data generation, private ANOVA test, and $p$-value calculation) 1,000 times for ($n$,$\epsilon$) pairs and recorded the proportion of iterations that report a $p$-value less than 0.05.  

For our simulated datasets, we selected databases ranging from ten records to one million records, and privacy parameters ranging from $\epsilon = \infty$ (non-private) to $\epsilon = 0.01$.  The choice of an ``acceptable'' $\epsilon$ value is a policy question: $\epsilon=.01$ is extremely conservative and allows for safe composition with many other queries, while $\epsilon=1$ provides meaningful privacy protection but might be too high for the comfort of some.

The power curves in Figure~\ref{fig:anova-est} quantify the number of records needed to consistently yield statistically significant results for ($n$,$\epsilon$) pairs.  In the non-private case (when $\epsilon=\infty$), databases with 100 records yield consistently significant $F$ values over the 1,000 runs.  As the value of $\epsilon$ decreases, more records are required to provide statistically significant results; with $\epsilon=1$ five to ten thousand data points are needed to frequently see significance, while nearly one million records are required when $\epsilon=0.01$. 

Whether these results are exciting or disappointing depends very much on one's reference point.  On the one hand, databases with several thousand values are extremely common, meaning that we can now conduct ANOVA tests on a wide variety of real-world databases while protecting privacy, something that was not previously possible.  On the other hand, the differentially private version is substantially less powerful than the traditional non-private version.

\begin{figure}
\centering
\includegraphics[width=\linewidth]{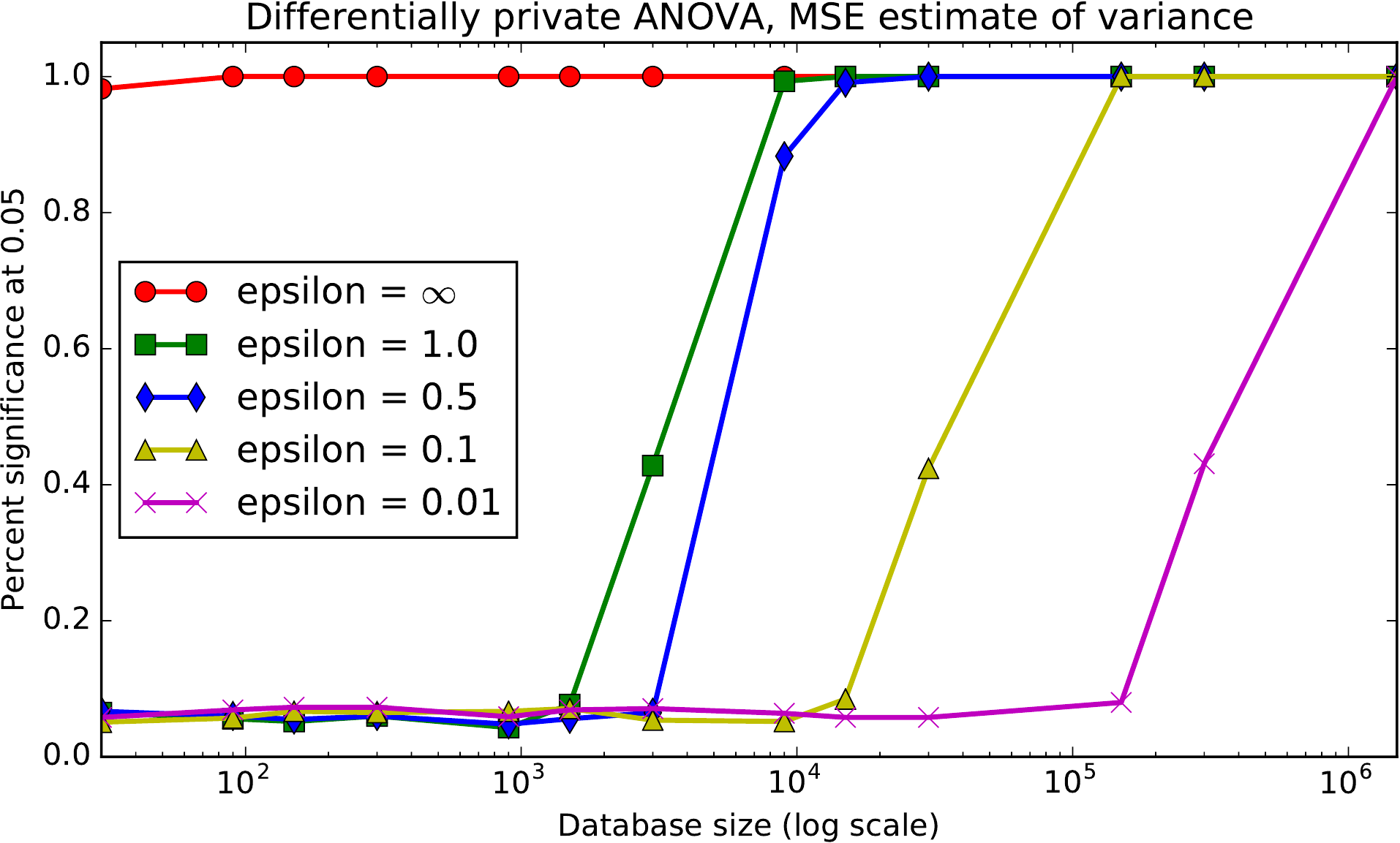}
\caption{Proportion of significant ($p<0.05$) $F$-ratios for databases of different sizes $n$ and different privacy parameters $\epsilon$.  The $p$-values were calculated by generating a distribution of $F$-ratios with variance estimated by $\widehat{SSE}/(n-k)$.}
\label{fig:anova-est}
\end{figure}

We also verified that $\widehat{SSE}/(n-k)$ is a sufficiently accurate estimate of $\sigma^2$ for the calculating $p$-values.  We did this by running the experiment a second time using the ground truth variance ($0.15^2$) in the null distribution generation; this would be impossible in the real world (where the true $\sigma^2$ is not known), but it is guaranteed to give accurate $p$-values. The result (Figure~\ref{fig:anova-real}) is extremely similar to our results using the estimated $\sigma^2$, giving us confidence that this estimate is acceptable for real world use.


\begin{figure}
\centering
\includegraphics[width=\linewidth]{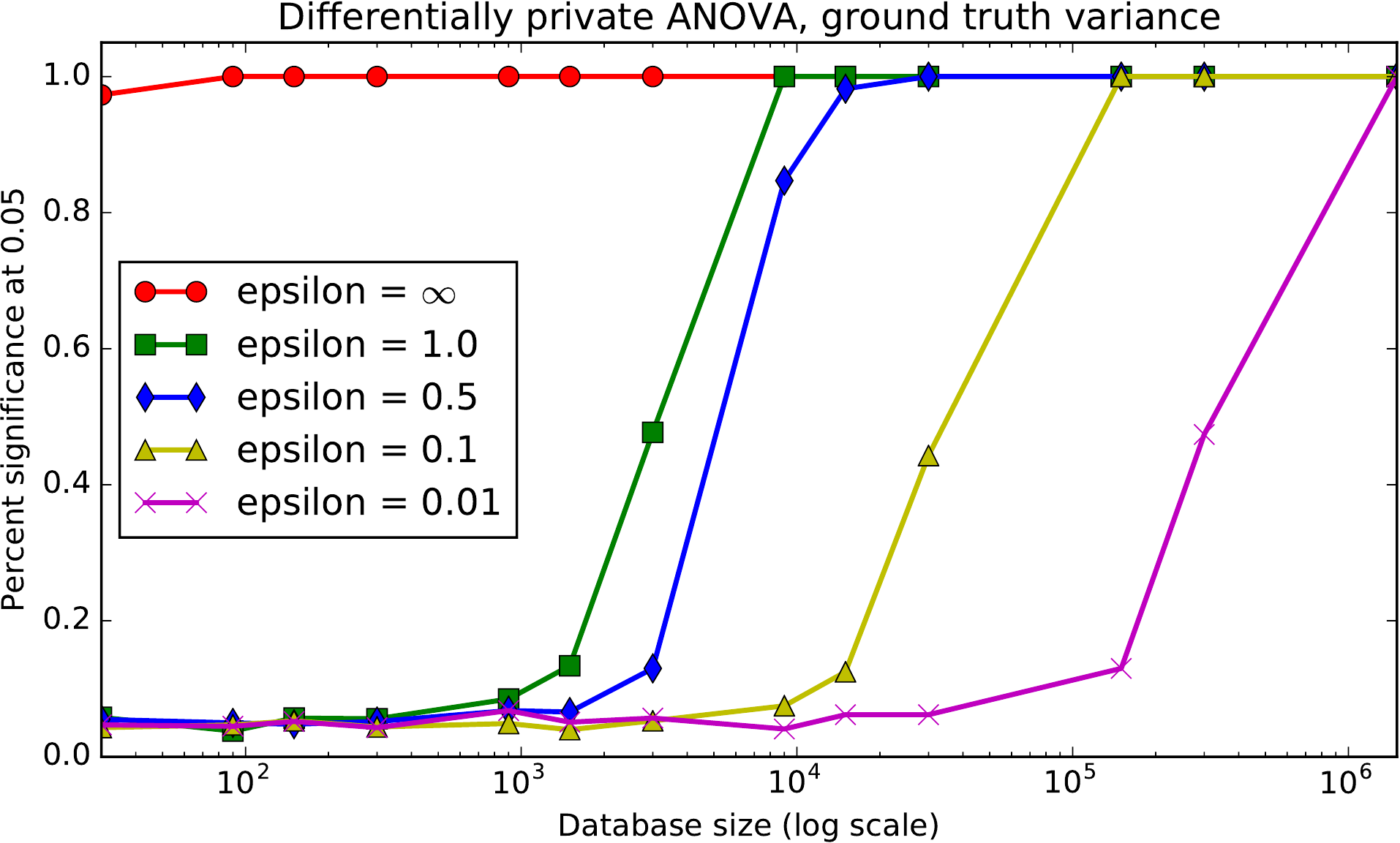}
\caption{Proportion of significant ($p<0.05$) $F$-ratios for databases of different sizes $n$ and different privacy parameters $\epsilon$.    The $p$-values were calculated by generating a distribution of $F$-ratios with the real (known) variance of $0.15^2$.}
\label{fig:anova-real}
\end{figure}


Next, we evaluated the differentially-private $F$-ratios on data with a smaller effect size. We generated six groups of values with means $[0.4,0.45,0.5,0.5,0.5,0.6]$ and standard deviation $0.2$, and again truncated values to be within $[0,1]$.  We computed the proportion of significant $p$-values as before.  The increased number of groups and the closer group means shift the power curves to the right (Figure~\ref{fig:anova-smaller-effect}), but even in this setting, a database size on the order of 10,000 would yield frequently significant $p$-values for privacy parameter $\epsilon=1.0$.
  

\begin{figure}
\centering
\includegraphics[width=\linewidth]{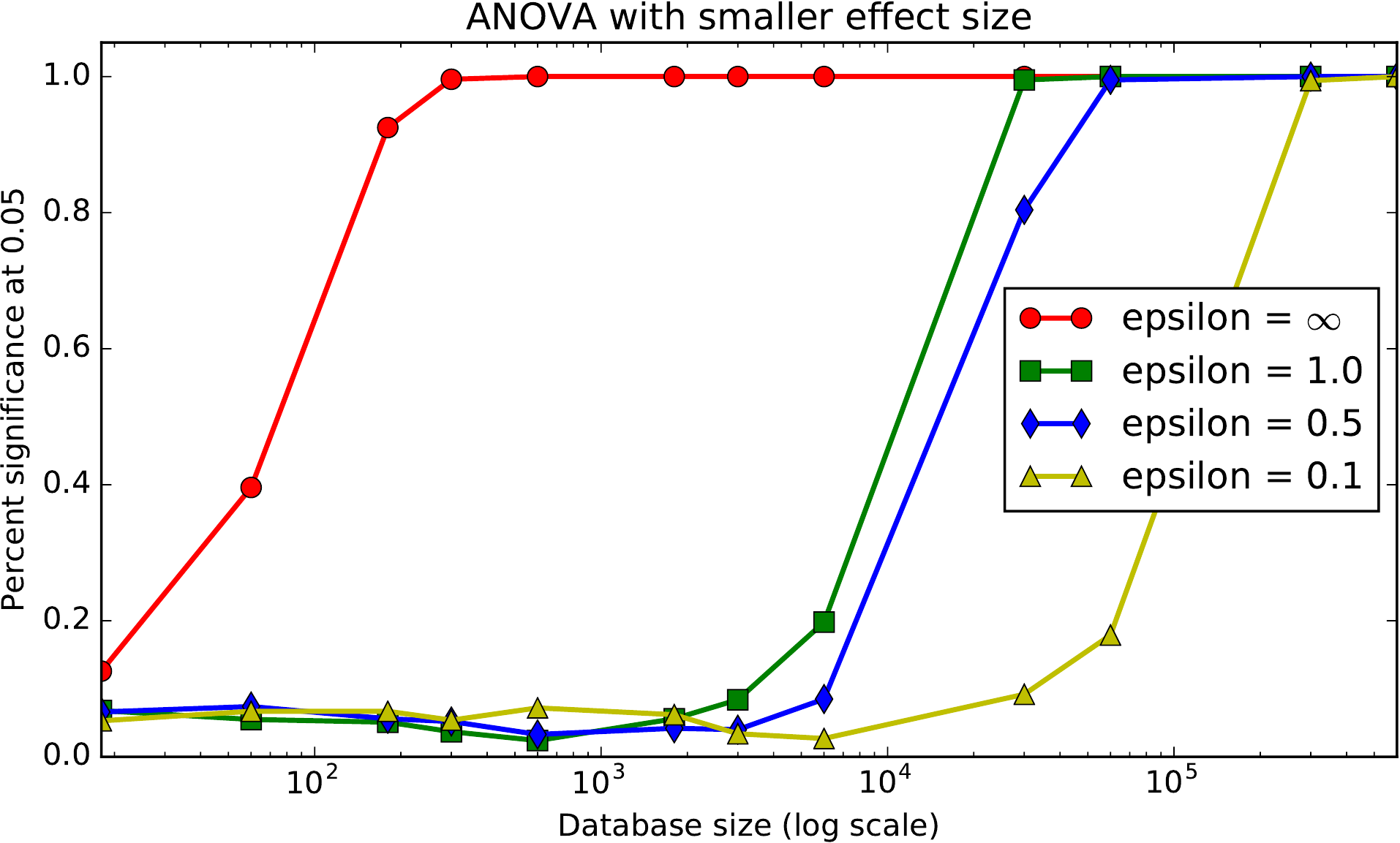}
\caption{Proportion of significant ($p<0.05$) $F$-ratios on data with a smaller effect size.}
\label{fig:anova-smaller-effect}
\end{figure}



\section{Conclusion and Future Work}

We show that for practically-sized databases, one can execute differentially private ANOVA tests with reasonable 
privacy guarantees and convincing results. 
We also quantify exactly how much data is needed to pick up an effect of a given size.  This is an important practical tool for data collection and analysis.  Perhaps more importantly, it gives a clear measurement of the efficacy of our algorithm in real world terms, allowing for easy comparison between this work and future attempts to improve the power of differentially private ANOVA testing.

This project is only a small step into what is a very large area for exploration.  There are other techniques for differentially private algorithms that might yield more powerful results.  We made some attempts to use the propose-test-release framework of Dwork and Lei \cite{dwork2009differential} or the smooth sensitivity framework of Nissim et al.~\cite{nissim2007smooth}.  These attempts gave us less effective algorithms than the one presented here, but it is certainly possible that further work could find a better algorithm.  

We hope that presenting algorithms for differentially private equivalents of frequently used statistical tests will make it easier for practitioners in other fields to make use of differentially private analysis and allow the useful study of data that was otherwise inaccessible to researchers.

\subsection*{Acknowledgment}

We thank Ira Globus-Harris for helpful comments on the presentation of our proofs.

\bibliographystyle{IEEEtran}
\bibliography{ANOVA}
%

\end{document}